\documentclass{article}

\usepackage{base}
\usepackage{macros}
\usepackage{minimacros}

\newcommand{\remove}[1]{}

\title{Improved Debordering of Waring Rank}
\author{Amir Shpilka
\thanks{This research was co-funded by the European Union by the European Union (ERC, EACTP, 101142020), the Israel Science Foundation (grant number 514/20) and  the Len Blavatnik and the Blavatnik Family Foundation. Views and opinions expressed are however those of the author(s) only and do not necessarily reflect those of the European Union or the European Research Council Executive Agency. Neither the European Union nor the granting authority can be held responsible for them.}
}
\date{}

\begin{document}

\maketitle

\begin{abstract}
    We prove that if a degree-$d$ homogeneous polynomial $f$ has border Waring rank $\underline{\mathrm{WR}}({f}) = r$, then its Waring rank is bounded by 
\[
{\mathrm{WR}}({f}) \leq d \cdot r^{O(\sqrt{r})}.
\]
This result significantly improves upon the recent bound ${\mathrm{WR}}({f}) \leq d \cdot 4^r$ established in [Dutta, Gesmundo, Ikenmeyer, Jindal, and Lysikov, STACS 2024], which itself was an improvement over the earlier bound ${\mathrm{WR}}({f}) \leq d^r$.
\end{abstract}

\section{Introduction}

Given a circuit class ${\cal C}$, its closure, $\overline{\cal C}$, is defined as the closure of the set of polynomials computable in ${\cal C}$.Specifically, this includes all polynomials that are limits, in the Zariski topology, of a converging sequence of polynomials computable in ${\cal C}$. Over the complex or real fields, this is equivalent to the coefficient vectors converging in the usual sense (i.e., with respect to the Euclidean topology). However, this notion also applies to arbitrary fields and can be defined algebraically.

In this paper, we study the closure of depth-$3$ powering circuit, denoted by $\SWS$. The output of a $\SWSrd{r}{d}$ circuit is a degree-$d$ homogeneous polynomial of the form 
\[f(\vecx)=\sum_{i=1}^{r}\ell_i(\vecx)^d.
\] 
where $\ell_i$ are linear forms (i.e., homogeneous linear polynomials).\footnote{We can allow representations of the form $\sum_{i=1}^{r}c_i\cdot \ell_i^d$, for scalars $c_i$, but over algebraically closed fields this does not change the complexity.} The Waring rank of a homogeneous polynomial of degree $d$ is defined as the minimal $r$ such that $f$ can be computed by a $\SWSrd{r}{d}$ circuit. It is also known as the \emph{symmetric tensor rank} of $f$. The \emph{border Waring rank} of a polynomial $f$, denoted $\bwr{f}$, is the minimal $r$ such that $f$ is in the closure of polynomials of Waring rank at most $r$.

Alder \cite{Alder84} (see also \cite[Appendix 20.6]{BCS-book97}) showed that $\bwr{f}=r$ if there exist $r$ linear functions $\ell_i \in \C(\epsilon)[\vecx]$ such that 
\[f=\lim_{\epsilon\to 0}\sum_{i=1}^{r}\ell_i^d.
\]
Equivalently, $\bwr{f}=r$ if there exist a degree-$d$ polynomial $g\in \C[\epsilon][\vecx]$, an integer $q$, and linear forms $\ell_i \in \C[\epsilon][\vecx]$ such that 
\[\epsilon^q f + \epsilon^{q+1}g=\sum_{i=1}^{r}\ell_i^d.
\]
This alternative definition generalizes to arbitrary characteristic fields.

Understanding the relationship between $\bwr{f}$ and $\wr{f}$ is a longstanding open problem. Forbes \cite{Forbes-talk} conjectures that $\overline{\SWS}=\SWS$. In other words, if $\bwr{f}=\poly(n,d)$, then $\wr{f}$ is also polynomially bounded. Ballico and Bernardi \cite{BallicoBernardi17} proposed a stronger conjecture, asserting that taking limits can save at most a factor of $d$.  Specifically, they conjectured
\[
\wr{f} \leq (\bwr{f} - 1) \cdot \deg(f).
\]   
This conjecture was verified for small values of $r$ ($r=3,4$, and when $d\geq 9$, also for $r=5$) \cite{LandsbergTeitler10,BallicoBernardi13,Ballico18}.

In \cite{DuttaGIJL24}, Dutta, Gesmundo, Ikenmeyer, Jindal, and Lysikov studied the general case and proved that if a polynomial $f$ of degree $d$ has $\bwr{f} = r$, then  
\[
\wr{f} \leq d \cdot 4^r.
\]  

We significantly improve upon the upper bound given in \cite{DuttaGIJL24}.  

\begin{theorem}\label{thm:main}  
Let $f \in \C[\vecx]_d$ be a homogeneous polynomial. If $\bwr{f} = r$, then  
\[
\wr{f} \leq d \cdot r^{10 \sqrt{r}}.
\]  
\end{theorem}

For more on debordering,  Waring rank, and related problems see \cite{bernardi2018hitchhiker,DuttaDwivediSaxena-demystifying21,DuttaGIJL24}.

\section{Preliminaries}

In this section, we introduce some notation and formally define the Waring rank and border Waring rank. We work over the field $\C$ of complex numbers. The space of homogeneous polynomials of degree $d$ in variables $x = (x_1,\ldots, x_n)$ is denoted by $\C[\vecx]_d$. We write $f \simeq g$ for $f, g \in \C(\epsilon)[\vecx]$ if $\lim_{\epsilon\to 0} f = \lim_{\epsilon\to 0} g$ (in particular, both limits must exist).

The projective space $\PP V$ is defined as the set of lines passing through the origin in $V$. For each nonzero $v \in V$, the corresponding line is denoted $[v] \in \PP V$, where $[v] = [w]$ if and only if $v = c\cdot w$ for some scalar $c\in \C$.

For integers $j \leq d$, we denote 
\[(d)_j=j!\cdot \binom{d}{j}=\prod_{i=0}^{j-1}(d-i),\]
where $(d)_j$ represents the falling factorial.

\subsection{Facts from \cite{DuttaGIJL24}}\label{sec:facts}

\begin{definition}[Waring Rank]\label{def:wr}
    Let \( f \in \mathbb{C}[\vecx] \) be a degree-\( d \) homogeneous polynomial.
    The \emph{Waring rank} of \( f \), denoted \( \wr{f} \), is the smallest integer \( r \) such that there exist homogeneous linear forms \( \ell_1, \ldots, \ell_r \) satisfying:
    \[
    f = \sum_{i=1}^{r} \ell_i^d.
    \]
\end{definition}

\begin{definition}[Border Waring Rank]\label{def:bwr}
    The \emph{border Waring rank} of \( f \), denoted \( \bwr{f} \), is the smallest \( r \) such that \( f \) can be expressed as a limit of a sequence of polynomials with Waring rank at most \( r \). 
\end{definition}

As shown in \cite{Alder84}, the next definition is equivalent to \autoref{def:bwr}. 

\begin{definition}[Border Waring Rank Decomposition]\label{def:bwr-alder}
    A \emph{border Waring rank decomposition} of a degree-\( d \) homogeneous polynomial \( f \in \mathbb{C}[\vecx]_d \) is an expression of the form:
    \[
    f = \lim_{\epsilon \to 0} \sum_{i=1}^{r} \ell_i^d,
    \]
    where \( \ell_1, \ldots, \ell_r \in \mathbb{C}(\epsilon)[\vecx]_1 \) are linear forms with coefficients rationally dependent on \( \epsilon \). The border Waring rank \( \bwr{f} \) is the smallest number \( r \) of summands in such a decomposition.
\end{definition}

A rational family of linear forms \( \ell \in \mathbb{C}(\epsilon)[\vecx]_1 \) always has a well-defined limit when viewed projectively.
Specifically, if \( \ell(\epsilon) \) is expanded as a Laurent series:
\[
\ell(\epsilon) = \sum_{i=q}^{\infty} \epsilon^i \ell_i, \quad \text{with } \ell_q \neq 0,
\]
then:
\begin{equation}\label{eq:local}
\lim_{\epsilon\to 0}[\ell(\epsilon)] = \lim_{\epsilon\to 0}\sum_{i=q}^{\infty}\epsilon^i \ell_{q+i} = [\ell_q].   
\end{equation}

A border Waring rank decomposition is called \emph{local} if, for all summands in the decomposition, this limit is the same.

\begin{definition}[Local Decomposition \cite{DuttaGIJL24}]
    Let \( f \in \mathbb{C}[\epsilon]_d \) be a degree-\( d \) homogeneous polynomial. A border Waring rank decomposition:
    \[
    f = \lim_{\epsilon \to 0} \sum_{i=1}^{r} \ell_i^d
    \]
    is called a \emph{local border decomposition} if there exists a linear form \( \ell \in \mathbb{C}[\vecx]_1 \) such that:
    \[
    \lim_{\epsilon \to 0} [\ell_i(\epsilon)] = [\ell] \quad \text{for all } i \in [r].
    \]
    The point \( [\ell] \in \mathbb{P}\mathbb{C}[\vecx]_1 \) is called the \emph{base} of the decomposition. 
    A local decomposition is called \emph{standard} if \( \ell_1 = c \cdot \epsilon^q \ell \) for some \( q \in \mathbb{Z} \) and \( c \in \mathbb{C} \). 
\end{definition}

The number of essential variables of a homogeneous polynomial \( f \) is the smallest integer \( m \) such that, after a linear change of coordinates, \( f \) can be expressed as a polynomial in \( m \) variables. Denote the number of essential variables of \( f \) by \( N(f) \).

\begin{lemma}[Lemma 4 of \cite{DuttaGIJL24}]
\label{lem:ess-var}
    For a homogeneous polynomial \( f \in \mathbb{C}[\vecx]_d \), we have \( N(f) \leq \bwr{f} \).
\end{lemma}

\begin{lemma}[Lemma 6 of \cite{DuttaGIJL24}]\label{lem:D:standard}
    If \( f \) has a local border decomposition, then it has a standard local border decomposition with the same base and the same number of summands. 
\end{lemma}

The following lemma can be proved by induction on the degree. It also follows from \autoref{cor:derivative}, as discussed in \autoref{rem:imply-D}.

\begin{lemma}[Lemma 7 of \cite{DuttaGIJL24}]\label{lem:D:divide}
    Suppose \( f \in \mathbb{C}[\vecx]_d \) has a local border decomposition with \( r \) summands based at \( [\ell] \). If \( d \geq r - 1 \), then:
    \[
    f = \ell^{d-r+1} \cdot g(\vecx),
    \]
    where \( g \) is a homogeneous polynomial of degree \( r - 1 \). 
\end{lemma}

The following lemma shows that when the degree is greater than the rank, the decomposition must be local (or consist of local decompositions).

\begin{lemma}[Lemma 10 of \cite{DuttaGIJL24}]\label{lem:local-partition}
    Let \( f \in \mathbb{C}[\vecx]_d \) be such that \( \bwr{f} = r \). If \( d \geq r - 1 \), then there exists a partition \( r = r_1 + \cdots + r_m \) such that \( f \) has a decomposition:
    \[
    f = \sum_{k=1}^{m} \ell_k^{d-r_k+1} \cdot g_k,
    \]
    where each \( \ell_k^{d-r_k+1} g_k \) has a local decomposition, with \( [\ell_k] \) as the base of the decomposition, and:
    \[
    \bwr{\ell_k^{d-r_k+1} \cdot g_k} \leq r_k.
    \]
\end{lemma}

The following simple lemma is implicit in \cite{DuttaGIJL24}.

\begin{lemma}[Perturbing a Variable by \( \epsilon \)]\label{lem:e-change}
    Let \( f \in \mathbb{C}[\epsilon][\vecx] \). If \( x_1 \simeq g \), where \( g \in \mathbb{C}(\epsilon)[\vecx] \), then:
    \[
    f(\vecx) \simeq f(g, x_2, \ldots, x_n).
    \]
\end{lemma}
\begin{proof}
    Let \( f_0 = \lim_{\epsilon \to 0} f(\vecx) \), where \( f_0 \in \mathbb{C}[\vecx] \). For some \( f_1 \in \mathbb{C}[\epsilon][\vecx] \), we can write \( f = f_0 + \epsilon f_1 \). Similarly, since \( \lim_{\epsilon \to 0} x_1 = \lim_{\epsilon \to 0} g \), we have \( g = x_1 + \epsilon g' \) for some \( g' \in \mathbb{C}[\epsilon][\vecx] \). Expanding \( f_0 \) monomial-wise and substituting \( g \) for \( x_1 \), it is straightforward to verify that:
    \[
    f_0(g, x_2, \ldots, x_n) = f_0(\vecx) + \epsilon f_2(\vecx),
    \]
    for some \( f_2 \in \mathbb{C}[\epsilon][\vecx] \). Thus, \[
    f(g, x_2, \ldots, x_n) = f_0(g, x_2, \ldots, x_n)+\epsilon f_1(g, x_2, \ldots, x_n)=  f_0(\vecx) + \epsilon f_2(\vecx)+\epsilon f_1(g, x_2, \ldots, x_n)\simeq f_0,
    \]
    as required.
\end{proof}

We will frequently use this lemma to simplify border Waring rank decompositions.\\

Additionally, if $\bwr{f} = r$ and $f = \lim_{\epsilon \to 0}\sum_{i=1}^{r}\ell_i^d$, where each $\ell_i(\vecx) \in \C(\epsilon)[\vecx]_1$, then for an integer $q$ such that $\epsilon^q \cdot \ell_i(\vecx) \in \C[\epsilon][\vecx]_1$ for all $i$, it holds that
\begin{equation*}
\epsilon^{qd}f + \epsilon^{qd+1}g = \sum_{i=1}^{r}(\epsilon^q \ell_i)^d,    
\end{equation*}
for some $g \in \C[\epsilon][\vecx]_d$.

Conversely, if for some integer $q$, polynomial $g \in \C[\epsilon][\vecx]_d$, and linear functions $\ell_i(\vecx) \in \C[\epsilon][\vecx]_1$, we have 
\begin{equation}\label{eq:nicer}
\epsilon^q f + \epsilon^{q+1}g = \sum_{i=1}^{r}\ell_i^d,
\end{equation}
then $\bwr{f} \leq r$. 

From this point onward, we will consider the representation in \eqref{eq:nicer} for polynomials $f$ with $\bwr{f} \leq r$.

\section{Improved debordering}

In this section, we provide the proof of \autoref{thm:main}. The proof begins by describing an $\epsilon$-perturbed diagonalization process. We consider homogeneous polynomials $f \in \C[\vecx]_d$ and assume, without loss of generality, that $N(f) = n \leq \bwr{f}$ (see \autoref{lem:ess-var}).

\begin{lemma}[Perturbed Diagonalization]
Let 
\[
\epsilon^q f + \epsilon^{q+1} g = \sum_{i=1}^{r} \ell_i^d
\]
be a decomposition of $f$, where $\ell_i(\vecx) \in \C[\epsilon][\vecx]_1$. Then, there exist:
\begin{itemize}
    \item a matrix $A = A_0 + \epsilon A_1 \in \C[\epsilon]^{r \times r}$, where $A_0 \in \C^{r \times r}$ is invertible,
    \item integers $0 = q_1 \leq q_2 \leq \ldots \leq q_n \leq q$,
    \item a permutation $\pi: [r] \to [r]$, and
    \item linear functions $L_1, \ldots, L_r$,
\end{itemize}
such that, for every $i \in [n]$ and $m \in [q_n]$, defining
\[
k_{i,m} = \argmax_{k \in [i-1]} \{ q_k \leq m \},
\]
the linear function $L_i(\vecx)$ satisfies the following:
\begin{equation}\label{eq:diagonal}
    L_i(\vecx) := \ell_{\pi(i)}(A\vecx) =  
    \begin{cases}  
    \sum_{m=0}^{q_i-1} \epsilon^m \sum_{k=1}^{k_{i,m}} c_{i,m,k} x_k + \epsilon^{q_i} x_i & \text{if } i \leq n, \\[5pt]
    \sum_{m=0}^{q} \epsilon^m \sum_{k=1}^{n} c_{i,m,k} x_k & \text{if } n < i \leq r.
    \end{cases}
\end{equation}
Furthermore, for some polynomial $\tilde{g} \in \C[\epsilon][\vecx]_d$, we have
\[
\epsilon^q f(A_0\vecx) + \epsilon^{q+1} \tilde{g}(\vecx) = \sum_{i=1}^{r} L_i^d,
\]
which is a border Waring rank decomposition of $f(A_0 \vecx)$. Moreover, if the original decomposition of $f$ was local, then the decomposition of $f(A_0\vecx)$ is also local, based at $x_1$.
\end{lemma}

To better understand this construction, consider the matrix $Q$ representing the $L_i$'s, where $Q_{i,j}$ is the linear form corresponding to the coefficient of $\epsilon^j$ in $L_i$. Explicitly:
\[
Q_{i,j} = 
\begin{cases} 
    \sum_{k=1}^{k_{i,j}} c_{i,j,k} x_k & \text{if } j < q_i, \\[5pt]
    x_i & \text{if } j = q_i, \\[5pt]
    0 & \text{if } j > q_i.
\end{cases}
\]
Thus, the first $n$ rows of the matrix $Q$ are in lower triangular form. Importantly, a variable $x_k$ does not appear in $L_1, \ldots, L_{k-1}$ and can only appear in columns $j \geq q_k$.

\begin{proof}[Proof of \autoref{lem:e-change}]
    First, observe that we can assume, without loss of generality, that no $\ell_i$ contains powers of $\epsilon$ larger than $q$. This simplification can be achieved by removing these higher-order terms from $\ell_i$, which would only modify $g$ without affecting the decomposition.

    Let us denote $C_j[\ell_i] \in \C[\vecx]_1$ as the coefficient of $\epsilon^j$ in $\ell_i$.

    In Algorithm~\ref{alg:diagonal}, we outline the process for constructing the matrix $A$. The algorithm begins by constructing a basis of linear functions. At each iteration, it identifies the smallest power of $\epsilon$ such that one of the remaining $\ell_i$ has a coefficient at that power which is a linear function linearly independent of all previously constructed basis elements. This $\ell_i$ is then removed from the set, and the identified linear function is added to the basis. This process repeats until all $\ell_i$ are processed.

    After this step, each $\ell_i$ is associated with a basis element and a corresponding power of $\epsilon$, representing the coefficient of that power in $\ell_i$ during the iteration. Additionally, the $\ell_i$s are re-indexed based on the order in which they were removed.

    Next, an invertible linear transformation is applied to ensure that each basis element corresponds to a variable. Specifically, under the new indexing, $x_i$ becomes the variable associated with the basis element derived from $\ell_i$. After the transformation, each $\ell_i$ takes the form:
    \[
    \ell_i = \tilde{\ell}_i(x_1, \ldots, x_{i-1}) + \epsilon^{q_i}x_i + \epsilon^{q_i+1}\ell'_i,
    \]
    where $\deg_{\epsilon}(\tilde{\ell}_i) < q_i$. Finally, we adjust $x_i$ by adding $\epsilon\ell'_i$, ensuring that:
    \[
    \ell_i = \tilde{\ell}_i(x_1, \ldots, x_{i-1}) + \epsilon^{q_i}x_i.
    \]

    \begin{algorithm}[H]
    \caption{Perturbed Diagonalization}\label{alg:diagonal}
        \begin{algorithmic}[1] 
            \State Set $I_0 = [r]$ and ${\cal L} = \emptyset$.
            \For{$k = 1 \ldots n$}
                \State Find the smallest power $j$ such that $C_j[\ell_i]$ is linearly independent of all elements in ${\cal L}$, for some $i \in I_{k-1}$.
                \State Set $q_k = j$, and let $i_k$ be the smallest index $i$ such that $C_{q_k}[\ell_i]$ is linearly independent of all the linear forms in ${\cal L}$.
                \State Update $I_k = I_{k-1} \setminus \{i_k\}$ and set $\pi(k) = i_k$.
                \State Define $\tilde{L}_k = C_{q_k}[\ell_{i_k}]$ and $\tilde{L}_{k,\epsilon} = \sum_{j = q_k+1}^{q} \epsilon^{j-q_k} C_j[\ell_{i_k}]$.
                \State Update ${\cal L} = {\cal L} \cup \tilde{L}_k$.
            \EndFor
            \State Complete $\pi$ to be a permutation on $[r]$.
            \State Define $A \in \C[\epsilon]^{r \times r}$ such that for all $k \in [r]$, $(\tilde{L}_k + \tilde{L}_{k,\epsilon})(A\vecx) = x_k$.
            \State Define $L_k(\vecx) := \ell_{\pi(k)}(A\vecx)$.
        \end{algorithmic}
    \end{algorithm}

    To verify the correctness of the constructed matrix $A$, observe that for each $k$, $\tilde{L}_{k,\epsilon}|_{\epsilon=0} = 0$, meaning that $\tilde{L}_k + \tilde{L}_{k,\epsilon} = \tilde{L}_k + O(\epsilon)$. By the definition of $q_k$, for every $m < q_k$, $C_m[\ell_i] \in \Span\{\tilde{L}_j \mid j < k\}$ for all $i \in [r]\setminus I_{k-1}$. Thus, after the variable transformation defined by $A$, only variables $x_1, \ldots, x_{k-1}$ appear in the coefficients of $\epsilon^m$ for $m < q_k$. Furthermore, the coefficient of $\epsilon^{q_k}$ in $L_k$ is precisely $x_k$, ensuring that \eqref{eq:diagonal} holds.

    Note that $A$ can be written as $A = A_0 + \epsilon A_1$, where $A_0 \in \C^{r \times r}$ is invertible because $\tilde{L}_k(A_0\vecx) = x_k$. Consequently, for some polynomial $\tilde{g} \in \C[\epsilon][\vecx]_d$, we have:
    \[
    f(A\vecx) + \epsilon \cdot g(A\vecx) = f(A_0\vecx) + \epsilon \cdot \tilde{g}(A_0\vecx).
    \]
    Therefore:
    \[
    \epsilon^q f(A_0\vecx) + \epsilon^{q+1} \tilde{g}(A_0\vecx) = \epsilon^q f(A\vecx) + \epsilon^{q+1}g(A\vecx) = \sum_{i=1}^{r}\ell_i(A\vecx)^d = \sum_{i=1}^{r}L_i^d.
    \]
The fact that the new decomposition remains local, provided the original decomposition was local, follows directly from the invertibility of $A_0$.
\end{proof}

\begin{corollary}\label{cor:derivative}
    Let $f$, $L_i$, and $\tilde{g}$ be as in \autoref{lem:e-change}. Then, for every $k\in [n]$ and $j\in [r-1]$, it holds that 
    \[
    \bwr{\frac{\partial^j f}{\partial x_k^j}} \leq r-k.
    \]
\end{corollary}
\begin{proof}
    By the definition of the $L_i$s, $x_k$ does not appear in $L_1, \ldots, L_{k-1}$ (see \eqref{eq:diagonal}). Therefore, using the notation of \autoref{lem:e-change}, we obtain:
    \begin{equation*}
        \epsilon^q \frac{\partial^j f(A_0\vecx)}{\partial x_k^j} + \epsilon^{q+1} \frac{\partial^j \tilde{g}(\vecx)}{\partial x_k^j} =
        \frac{\partial^j}{\partial x_k^j} \left( \epsilon^q f(A\vecx) + \epsilon^{q+1}g(A\vecx)\right)
        = \frac{\partial^j}{\partial x_k^j} \left( \sum_{i=1}^{r} L_i^d \right) 
        = (d)_j \cdot \sum_{i=k+1}^{r} \left( \frac{\partial L_i}{\partial x_k} \right)^j \cdot L_i^{d-j}.
    \end{equation*}
    Since $A_0$ is invertible, it follows that 
    \[
    \bwr{\frac{\partial^j f}{\partial x_k^j}} \leq r-k.
    \]
\end{proof}

\begin{remark}\label{rem:strengthened}
    The conclusion of \autoref{cor:derivative} can be strengthened to 
    \[
    \bwr{\frac{\partial^j f}{\partial x_k^j}} \leq r-k-j+1,
    \]
    as after each derivative, we can re-diagonalize and conclude that each successive derivative, not taken with respect to $x_1$, reduces the rank further.
\end{remark}

\begin{remark}\label{rem:imply-D}
    We note that \autoref{cor:derivative} implies \autoref{lem:D:divide}, as it shows that taking a derivative with respect to any variable other than $x_1$ reduces the Waring rank. Consequently, any monomial can contain at most $r-1$ other variables.
\end{remark}

We now give the proof of \autoref{thm:main}. 
\begin{proof}[Proof of \autoref{thm:main}]

    The proof proceeds by induction on $r$. 
    For $1 < r \leq 100$, the result of \cite{DuttaGIJL24} implies that $\wr{f} \leq 4^r < r^{10\sqrt{r}}$. Hence, we assume from now on that $r \geq 100$.

    Let $\epsilon^q f + \epsilon^{q+1} g = \sum_{i=1}^{r} \ell_i^d$
    be a border Waring rank decomposition of $f$, where $\ell_i(\vecx) \in \C[\epsilon][\vecx]_1$. 
    By applying \autoref{lem:e-change}, we can assume without loss of generality that the $\ell_i$s are in diagonal form, as described in the lemma.  

    We handle two separate cases. The first is when $d \geq r-1$, and the second is when $d < r-1$. 

    \paragraph{The case $d \geq r-1$.} \sloppy
    From \autoref{lem:local-partition}, there exists a partition $r = r_1 + \ldots + r_m$ such that $f$ has a decomposition
    \[
    f = \sum_{k=1}^{m} \ell_k^{d-r_k+1} \cdot g_k,
    \]
    where $\ell_k^{d-r_k+1} g_k$ has a local decomposition, with $[\ell_k]$ being the base of the decomposition, and $\bwr{\ell_k^{d-r_k+1} \cdot g_k} \leq r_k$. 

    Since $d \cdot \sum_{k=1}^{m} r_k^{10\sqrt{r_k}} \leq d \cdot r^{10\sqrt{r}}$, it suffices to prove \autoref{thm:main} for local decompositions when $d \geq r-1$.

    Assume that $f$ has a local border Waring rank decomposition, based in $x_1$, as in the conclusion of \autoref{lem:e-change}.

    Let $Y = \{x_1, \ldots, x_{\lfloor 10\sqrt{r} \rfloor}\}$ and $Z = \{x_{\lfloor 10\sqrt{r} \rfloor+1}, \ldots, x_n\}$. For convenience, rename the variables in $Z$ as $Z = \{z_1, \ldots, z_m\}$ for $m = n - \lfloor 10\sqrt{r} \rfloor$. By \autoref{lem:D:divide}, we have the following representation of $f$: 
    \begin{equation}\label{eq:f-representation}
    f = x_1^{d-r+1} \left(f_0(Y) + \sum_{i=1}^{m} \sum_{k=1}^{r-1} z_i^k \cdot g_{i,k}(Y, z_{i+1}, \ldots, z_m)
    \right).    
    \end{equation}
    In other words, we first consider monomials involving only the $Y$ variables. Then, each other monomial contains one or more variables from $Z$, and we group these monomials according to the minimal $i$ and then the maximal $k$ such that $z_i^k$ divides them.

    Clearly, $f_0$ is a polynomial of degree $r-1$ in $\lfloor 10\sqrt{r} \rfloor$ variables, and hence its Waring rank satisfies
    \[
    \wr{f_0} \leq \binom{\lfloor 10\sqrt{r} \rfloor + r - 3}{r-2} = \binom{\lfloor 10\sqrt{r} \rfloor + r - 3}{\lfloor 10\sqrt{r} \rfloor - 1} \leq \left(\frac{e(r + \lfloor 10\sqrt{r} \rfloor - 3)}{10\sqrt{r} - 1}\right)^{\lfloor 10\sqrt{r} \rfloor - 1} < (5 \cdot r)^{5\sqrt{r}}.
    \]
    Consequently, 
    \[
    \wr{x_1^{d-r+1} \cdot f_0} < d \cdot (5 \cdot r)^{5\sqrt{r}}.
    \]
 
    Next, observe that $x_1^{d-r+1} \cdot g_{i,k}$ can be obtained by taking $k$ derivatives of $f$ with respect to $z_i$, setting $z_1 = \ldots = z_i = 0$, and multiplying the result by $k!$. 
    
    From \autoref{cor:derivative}, we conclude that $\bwr{x_1^{d-r+1} \cdot g_{i,k}} \leq r - \lfloor 10\sqrt{r} \rfloor - i$, and clearly it is a polynomial on at most $n - i$ variables. The induction hypothesis implies that
    \[
    \wr{x_1^{d-r+1}\cdot  g_{i,k}} \leq d \cdot  (r - \lfloor 10\sqrt{r} \rfloor - i)^{10\sqrt{r - \lfloor 10\sqrt{r} \rfloor - i}}.
    \]
    Hence,
    \[
    \wr{x_1^{d-r+1}\cdot z_i^k\cdot  g_{i,k}} \leq r\cdot d \cdot  (r - \lfloor 10\sqrt{r} \rfloor - i)^{10\sqrt{r - \lfloor 10\sqrt{r} \rfloor - i}}.
    \]
    It follows that    
    \begin{align*}
    \wr{f} &< d \cdot (5 \cdot r)^{5\sqrt{r}} + d \cdot \sum_{i=1}^{m} \sum_{k=1}^{r-1} r\cdot (r - \lfloor 10\sqrt{r} \rfloor - i)^{10\sqrt{r - \lfloor 10\sqrt{r} \rfloor - i}} \\
    &< d \cdot (5 \cdot r)^{5\sqrt{r}} + d \cdot r^2\cdot \sum_{i=1}^{m} (r - \lfloor 10\sqrt{r} \rfloor - i)^{10\sqrt{r - \lfloor 10\sqrt{r} \rfloor - i}} \\
    &< d \cdot (5 \cdot r)^{5\sqrt{r}} + d \cdot r^3 \cdot r^{10(\sqrt{r} - 5)} \\
    &< d \cdot r^{10\sqrt{r}}.
    \end{align*}

    \paragraph{The case $d < r-1$.}
    Assume that $f$ has a border Waring rank decomposition as in the conclusion of \autoref{lem:e-change}. As before, set $Y = \{x_1, \ldots, x_{\lfloor 10\sqrt{r} \rfloor}\}$ and $Z = \{x_{\lfloor 10\sqrt{r} \rfloor+1}, \ldots, x_n\}$, and rename the variables in $Z$ as $Z = \{z_1, \ldots, z_m\}$ for $m = n - \lfloor 10\sqrt{r} \rfloor$. Using the same reasoning as before, we conclude that 
    \begin{align*}
    \wr{f} &< d \cdot (5 \cdot r)^{5\sqrt{r}} + d \cdot r\cdot \sum_{i=1}^{m} \sum_{k=1}^{r-1} (r - \lfloor 10\sqrt{r} \rfloor - i)^{10\sqrt{r - \lfloor 10\sqrt{r} \rfloor - i}} \\
    &< d \cdot r^{10\sqrt{r}}. \qedhere
    \end{align*}
\end{proof}

\section*{Acknowledgement}
I  thank Prerona Chatterjee, Shir Peleg and Ben lee Volk for valuable discussions on the problem.

\bibliographystyle{alpha}
\bibliography{main}

\newcommand{\etalchar}[1]{$^{#1}$}
\begin{thebibliography}{BCC{\etalchar{+}}18}

\bibitem[Ald84]{Alder84}
Alexander Alder.
\newblock {\em Grenzrang und Grenzkomplexität aus algebraischer und topologischer Sicht}.
\newblock Phd thesis, Universität Zürich, 1984.

\bibitem[Bal18]{Ballico18}
Edoardo Ballico.
\newblock On the ranks of homogeneous polynomials of degree at least 9 and border rank 5.
\newblock {\em Note Mat.}, 38(2):55--92, 2018.

\bibitem[BB13]{BallicoBernardi13}
Edoardo Ballico and Alessandra Bernardi.
\newblock Stratification of the fourth secant variety of {V}eronese varieties via the symmetric rank.
\newblock {\em Adv. Pure Appl. Math.}, 4(2):215--250, 2013.

\bibitem[BB17]{BallicoBernardi17}
Edoardo Ballico and Alessandra Bernardi.
\newblock Curvilinear schemes and maximum rank of forms.
\newblock {\em Matematiche (Catania)}, 72(1):137--144, 2017.

\bibitem[BCC{\etalchar{+}}18]{bernardi2018hitchhiker}
Alessandra Bernardi, Enrico Carlini, Maria~Virginia Catalisano, Alessandro Gimigliano, and Alessandro Oneto.
\newblock The hitchhiker guide to: Secant varieties and tensor decomposition.
\newblock {\em Mathematics}, 6(12):314, 2018.

\bibitem[BCS97]{BCS-book97}
Peter B\"urgisser, Michael Clausen, and M.~Amin Shokrollahi.
\newblock {\em Algebraic complexity theory}, volume 315 of {\em Grundlehren der mathematischen Wissenschaften [Fundamental Principles of Mathematical Sciences]}.
\newblock Springer-Verlag, Berlin, 1997.
\newblock With the collaboration of Thomas Lickteig.

\bibitem[DDS22]{DuttaDwivediSaxena-demystifying21}
Pranjal Dutta, Prateek Dwivedi, and Nitin Saxena.
\newblock Demystifying the border of depth-3 algebraic circuits.
\newblock In {\em 2021 {IEEE} 62nd {A}nnual {S}ymposium on {F}oundations of {C}omputer {S}cience---{FOCS} 2021}, pages 92--103. IEEE Computer Soc., Los Alamitos, CA, [2022] \copyright 2022.

\bibitem[DGI{\etalchar{+}}24]{DuttaGIJL24}
Pranjal Dutta, Fulvio Gesmundo, Christian Ikenmeyer, Gorav Jindal, and Vladimir Lysikov.
\newblock Fixed-parameter debordering of waring rank.
\newblock In Olaf Beyersdorff, Mamadou~Moustapha Kant{\'{e}}, Orna Kupferman, and Daniel Lokshtanov, editors, {\em 41st International Symposium on Theoretical Aspects of Computer Science, {STACS} 2024, March 12-14, 2024, Clermont-Ferrand, France}, volume 289 of {\em LIPIcs}, pages 30:1--30:15. Schloss Dagstuhl - Leibniz-Zentrum f{\"{u}}r Informatik, 2024.

\bibitem[For16]{Forbes-talk}
Michael Forbes.
\newblock Some concrete questions on the border complexity of polynomials.
\newblock Presentation given at the Workshop on Algebraic Complexity Theory (WACT), 2016.

\bibitem[LT10]{LandsbergTeitler10}
J.~M. Landsberg and Zach Teitler.
\newblock On the ranks and border ranks of symmetric tensors.
\newblock {\em Found. Comput. Math.}, 10(3):339--366, 2010.

\end{thebibliography}

\end{document}